\pdfoutput=1

\documentclass{article}

\usepackage{arxiv}

\usepackage{times}

\usepackage{multicol}
\usepackage[bookmarks=true]{hyperref}

\usepackage[pdftex]{graphicx}
\usepackage{epstopdf}
\usepackage{amsmath} 
\usepackage{amssymb} 
\usepackage{amsthm}
\usepackage[noend]{algpseudocode}
\usepackage{algorithm}
\usepackage{graphicx}
\usepackage{float}
\usepackage{setspace}
\usepackage{subfigure}
\usepackage{color}
\usepackage{cases}
\usepackage{cite}
\usepackage{ mathrsfs }
\usepackage{comment} 
\usepackage{tabularx}
\usepackage{verbatim}
\usepackage{caption}
\usepackage{comment}
\usepackage{fourier}
\usepackage[bb=ams, cal=cm, scr=boondox, frak=euler]{mathalpha}

\newtheorem{theorem}{Theorem}[section]
\newtheorem{lemma}[theorem]{Lemma}

\newtheorem{corollary}[theorem]{Corollary}
\newtheorem{problem}[theorem]{Problem}
\newtheorem{assumption}[theorem]{Assumption}
\newtheorem{hypothesis}[theorem]{Hypothesis}

\newtheorem{definition}[theorem]{Definition}

\title{Dependent Reachable Sets for the Constant Bearing Pursuit Strategy}

\author{
Venkata Ramana Makkapati\\
Honda Aircraft Company\\
Greensboro, NC, USA\\
\texttt{vmakkapati3@gmail.com}\\
\And
Tulasi Ram Vechalapu\\
Technion - Israel Institute of Technology\\
Haifa, Israel\\
\texttt{tulasiram217@campus.technion.ac.il}\\
\And
Vinodhini Comandur\\
University of Colorado Boulder\\
Boulder, CO, USA\\
\texttt{vinodhini.comandur@colorado.edu}\\
\And
Seth Hutchinson\\
Northeastern University\\
Boston, MA, USA\\
\texttt{s.hutchinson@northeastern.edu}\\
}

\begin{document}

\maketitle

	
\begin{abstract}
This paper introduces a novel reachability problem for the scenario involving two agents, where one agent follows another agent using a feedback strategy.
The geometry of the reachable set for an agent, termed \emph{dependent reachable set}, is characterized using the constant bearing pursuit strategy as a case study.
Key theoretical results are presented that provide geometric bounds for the associated dependent reachable set.
Simulation results are presented to empirically establish the shape of the dependent reachable set.
In the process, an original optimization problem is formulated and analyzed for the constant bearing pursuit strategy.
\end{abstract}


	
	
\section{Introduction}
\label{sec:intro}
Reachability analysis is a key formal verification tool for autonomous and cyber-physical systems \cite{althoff2014online, FAURE2022272, herbert2017fastrack}.
It provides provable safety guarantees and enables fault-tolerant control under uncertainties~\cite{chen2022reach, jeremy2011app}.
Forward reachability analysis computes the \emph{forward reachable set}, which contains all states a system (or agent) can reach over time, given the initial state and the set of allowable control inputs. 
In contrast, backward reachability analysis identifies the \emph{backward reachable set}, the set of states from which a given target state can be reached at a specified time.
In general, it is challenging to analytically compute the forward or backward reachable sets for nonlinear systems, even for most constrained linear systems \cite{althoff2021set}.
Consequently, several approaches, including those based on optimal control theory \cite{mitchell2005time, chen2018hamilton, sun2017pursuit}, barrier functions \cite{kong2018reachable}, zonotopes \cite{Li11062023, Zhang12032021}, and trajectory sensitivity analysis \cite{ramdani2008reachability, chuchu2016locally}, were developed to numerically estimate reachable sets for linear and nonlinear systems.
In recent years, there has been an increased focus on reachability analysis for multi-agent systems \cite{thapliyal2024algorithm, wang2021formal} and neural network controllers \cite{gates2023scalable, everett2021reach}. 

In the context of formal verification for multi-agent systems, where one agent (referred to as \emph{dependent agent}) follows another agent (referred to as \emph{independent agent}) using a feedback strategy, it is important to characterize the reachable set of the dependent agent for a given strategy, termed \emph{dependent reachable set} (DRS) in this paper.
This problem is particularly relevant in defense and security applications, where it is critical to estimate the worst-case scenario for the dependent agent by evaluating the set of states that the agent can reach when following an adversary using a feedback strategy.
Conversely, to plan for contingencies in strategic engagements, the independent agent must identify the set of locations to which it can lead the dependent agent.
To this end, in this paper, DRS for two-dimensional problems involving two agents is formalized, which can be easily extended to systems with more than two agents and to three-dimensional environments.

Regarding the applicability of DRS, Comandur et al. noted that for a class of pursuit-evasion games involving deception, estimating such reachable sets for the deceptive player (which in our paper corresponds to the independent agent) is essential to assess the relevance of its strategy \cite{comandur2024desensitization}.
Their work considered the pure pursuit strategy for the dependent agent.
However, the lack of closed-form solutions in the existing literature on pure pursuit strategy~\cite{Shneydor, makkapati2018optimal} and our preliminary analysis suggest that the geometry of the associated DRS cannot be expressed analytically in closed form.
We therefore focus on the constant bearing pursuit strategy, which is widely used in missile guidance problems \cite{Shneydor}.
The constant bearing pursuit strategy is also known as ``parallel navigation" in the robotics literature \cite{strydom2015biologically}, and was applied to mobile robots to reach a moving target \cite{belkhouche2007parallel, rafie2010time,makkapati2019optimal}.
The geometric elegance of constant bearing pursuit is supported by its connection to the Apollonius circle~\cite{makkapati2018optimal}, a property also evident in our results.

A basic approach to computing the DRS at a fixed time $t$ is to first determine the set of states reachable by the dependent agent for each possible state of the independent agent at the same time $t$, and then perform a union of these sets.
Note that the independent agent may reach a certain state via multiple trajectories, allowing the dependent agent to reach multiple corresponding states.
In this regard, an optimization problem is formulated to identify the maximum and minimum relative distances that the dependent agent can achieve, given the initial conditions and current position of the independent agent, at a specified time, while following the constant bearing pursuit strategy.
Although an analytical solution to this optimization remains elusive, simulations reveal interesting geometric properties linked to ellipses.

To the best of the authors' knowledge, the concept of DRS is novel and unexplored in the multi-agent systems literature.
The contributions of this paper are as follows.
\begin{enumerate}
    \item The concept of DRS is formalized.

    \item The geometry of DRS for the constant bearing pursuit strategy is characterized using theoretical and simulation results.

    \item A theoretical result that establishes the link between the DRS of the constant bearing pursuit strategy and the Apollonius circle is presented.

    \item An original optimization problem for maximum and minimum relative distances under the constant bearing pursuit strategy is formulated and analyzed.

    \item A new property of ellipses is observed in the context of the constant bearing pursuit strategy using empirical results.
\end{enumerate}

The remainder of the paper is organized as follows. Section~\ref{sec:prob} presents the problem formulation and the underlying assumptions. 
Section~\ref{sec:reach_set} reviews the theory of reachable sets. Section~\ref{sec:main_result} presents key theoretical and simulation results describing the geometry of the DRS for constant bearing pursuit. Section~\ref{sec:optimize} analyzes the optimization problem to find the closest and the farthest points from the independent agent that the dependent agent can reach while employing the constant bearing pursuit strategy. Section~\ref{sec:conclude} concludes the paper and outlines directions for future work.

	
\section{Problem Statement}
\label{sec:prob}

Consider the standard two-dimensional scenario involving two agents (denoted using the subscripts $I$ and $D$) moving in the Cartesian plane at a constant speed by controlling their respective heading angles.
The agents' dynamics can be expressed as
\begin{align}
    \dot{\mathbf{x}}_i(t) &= \mathbf{u}_i(t), \quad \mathbf{x}_i(0) = \mathbf{x}_i^0, \quad \|\mathbf{u}_i(t)\|_2 = v_i, \label{eq:agent_Dyn}
\end{align}
where $\mathbf{x}_i(t) = [x_i(t),~y_i(t)]^\top \in \mathbb{R}^2$, $i \in \{I,D\}$ denotes the agent's position at time $t\geq 0$, and $\|\cdot\|_2$ denotes the 2-norm.
Here, $x_i$ is the horizontal coordinate (also referred to as $x$-coordinate) and $y_i$ is the vertical coordinate (also referred to as $y$-coordinate).
Similarly, $\mathbf{u}_i(t)$ is the agent's instantaneous velocity vector (control input) at time $t$, $\mathbf{x}_i^0$ is the initial position at time $t=0$, and $v_i$ is the agent's speed, which is a constant.
Given $\mathcal{U}_i$ is the set of all piecewise continuous functions in time $t$ for which the range is the set $\mathbb{U}_i = \{\mathbf{u} \in \mathbb{R}^2 : \|\mathbf{u}\|_2 = v_i\}$, the control function $\mathbf{u}_i(.) \in \mathcal{U}_i$.
At time $t$, the heading angle of an agent given the control vector $\mathbf{u}_i(t)$ is denoted as $\psi_i(t) \in (-\pi,\pi]$, where $\mathbf{u}_i(t) = [v_i\cos\psi_i(t),~v_i\sin\psi_i(t)]^\top$.
The heading angle is measured from the horizontal $x$-axis, with the counterclockwise direction being positive.
In this paper, we examine instances involving two agents where one, the dependent agent  (indicated using the subscript $D$)
\emph{follows} the other, the independent agent (indicated using the subscript $I$), 
as part of its mission or task, using the constant bearing pursuit strategy.

Without loss of generality, we set $\mathbf{x}^0_{D} = [0,0]^\top$ (origin) and $\mathbf{x}^0_I = [a,0]^\top$.
Consequently, the initial line of sight (LOS) is parallel to the horizontal axis.
Since the dependent agent follows the constant bearing pursuit strategy, it chooses a heading angle at a given time instant such that the LOS does not rotate, i.e., the LOS is parallel to the horizontal axis for all time, and the relative distance between the agents does reduce \cite{makkapati2018optimal}.
Therefore,
\begin{align}
    \psi_D(t) = \sin^{-1} \left(\dfrac{v_I}{v_D}\sin \psi_I(t)\right). \label{eq:const_bear}
\end{align}
In the remainder of this paper, we consider only the case where $\mathbf{u}_D(t) = [v_D\cos\psi_D(t),~v_D\sin\psi_D(t)]^\top$ for $\psi_D(t)$ given in (\ref{eq:const_bear}).

\begin{assumption} \label{assume:speed}
    The speed of the dependent agent is greater than that of the independent agent, i.e. $v_D > v_I$.
\end{assumption}

Note that $v_D \geq v_I$ is a necessary and sufficient condition for $\psi_D(t)$ to exist, given any $\psi_I(t) \in (-\pi,\pi]$, as can be inferred from (\ref{eq:const_bear}).
Furthermore, given $v_D\geq v_I$ and $\psi_I(t) \in (-\pi,\pi]$, $\psi_D(t) \in \left[-\dfrac{\pi}{2},\dfrac{\pi}{2}\right]$.
In this paper, $v_D > v_I$ is considered for presenting theoretical and empirical results.
The analysis corresponding to the limiting case of $v_D = v_I$ is presented in Section \ref{subsec:limit_case}.

Given the dynamics of the independent agent per (\ref{eq:agent_Dyn}), its trajectory $\mathbf{x}_I(t),~t\geq0$ is a function of the control input $\mathbf{u}_I \in \mathcal{U}_I$.
Since $\mathbf{u}_D(t)$ is a function of $\mathbf{u}_I(t)$, per (\ref{eq:const_bear}), for every trajectory of the independent agent, there is a corresponding trajectory of the dependent agent.
Now, the problem statement examined in this paper is presented below.

\begin{problem}\label{prob:drs}
    For the set of all feasible trajectories of the independent agent, determine the corresponding set, $\mathcal{D}(t)$, of positions (or states) that will be reached by the dependent agent following the constant bearing pursuit strategy.
\end{problem}

In other words, given that the dependent agent is committed to the constant bearing pursuit strategy, what is the set of points to which the independent agent can drive the dependent agent?
When $v_D > v_I$, every possible trajectory of the independent agent will result in a \emph{capture} in finite time by the dependent agent that follows the constant bearing pursuit strategy.
Here, capture refers to the event $\|\mathbf{x}_D(t)-\mathbf{x}_I(t)\|_2 \approx 0$.
In this regard, the paper considers that the set $\mathcal{D}(t)$ comprises only those points that correspond to the active pursuit trajectories at a given time.
By active pursuit trajectories, we mean those trajectories of the dependent agent that have yet to result in capture for the corresponding trajectory of the independent agent.
The maximum capture time $t_c = a/(v_D-v_I)$ corresponds to an instance in which the independent agent chooses the control input $\mathbf{u}_I(t) = [v_I,0]^\top$ for all time until capture.
In the pursuit-evasion literature, this chosen control input of the independent agent corresponds to what is known as \emph{pure evasion}, which guarantees the maximum capture time.
Note that as $v_D \rightarrow v_I$, $t_c \rightarrow \infty$.

In this paper, a point from the 2D plane is either declared as a column vector of the form $\mathbf{x} = [x,y]^\top$, or denoted in the Cartesian form $X\left(x,y\right)$ as an ordered pair. 
The following section presents important results for the reachable set of the independent agent, which will be used to characterize the set $\mathcal{D}(t)$ for the dependent agent.


\section{Reachable Sets}
\label{sec:reach_set}

The standard definition for the reachable set of the independent agent at time $t \geq 0$ is given below.

\begin{definition}
The reachable set of the independent agent, $\mathcal{R}_I(\mathbf{x}_I^0,t)$, at time $t \geq 0$ with the initial state at $\mathbf{x}_I^0$ is the set of all points that the agent can reach at time $t$:\label{def:reachable_set}
\begin{align}
    \mathcal{R}_I(\mathbf{x}_I^0,t) = \Bigg\lbrace\mathbf{x}\in\mathbb{R}^2 : \exists~\mathbf{u}_I \in \mathcal{U}_I,~\mathbf{x} = \mathbf{x}_I^0 + \int_0^t \mathbf{u}_I(\tau) \mathrm{d} \tau \Bigg\rbrace.
\end{align}
\end{definition}

Hereafter, the dependence of the reachable set on the initial condition, which is apparent, is dropped for brevity, and the reachable set at time $t$ is denoted as $\mathcal{R}_I(t)$.
A closed-form expression for the reachable set of the independent agent, per the dynamics in (\ref{eq:agent_Dyn}), can be obtained using the following lemmas. 
Though the analytical solution for the reachable set is common knowledge, to the best of the authors' knowledge, a formal proof for the same is not found in the existing literature. 
In addition, the theoretical results for the dependent reachable set, presented in Section \ref{sec:main_result}, are developed based on the following lemmas.

\begin{lemma}\label{lemma:1D_reach}
    The reachable set of an agent in a one-dimensional environment with the dynamics
    \begin{align}
        \dot{x}(t) = w(t), \quad x(0) = 0, \label{eq:1D_Dyn}
    \end{align}
    at time $t=T \in [0, \infty)$ is $[-T, T]$, where $x(t) \in \mathbb{R}$, the control function $w \in \mathcal{W}$, 
    and $\mathcal{W}$ is the set of all piecewise continuous functions in time $t\in [0,\infty)$ with range $\{-1, 1\}$.
\end{lemma}

\begin{proof}
    Consider the single-switching control function
    \begin{align}
    w_s(t) = \begin{cases}
    1, \quad \text{if } 0\leq t < t_s,\\
    -1, \quad \text{if } t_s \leq t \leq T,
    \end{cases} \label{eq:switch_func}
    \end{align}
    where $t_s \in (0, T)$ is the switching time. 
    Any point $x_s \in (-T,T)$ can be reached with the control function in (\ref{eq:switch_func}) using switching time
    \begin{align}
        t_s &= \dfrac{x_s + T}{2}.
    \end{align}
    It can be observed that the switching times vary between 0 and $T$ to span all the points in the set $(-T,T)$. 
    Since the agent moves at a unit speed, the maximum distance it can reach from its initial point at time $T$ is $T$.
    Therefore, in the one-dimensional case, the reachable set at time $T$ is $[-T,T]$.
\end{proof}

\begin{lemma}\label{lemma:2D_reach}
    The reachable set of the independent agent, per the dynamics in (\ref{eq:agent_Dyn}), at time $T \in [0, \infty)$ is the circle with its center at $\mathbf{x}_I^0$ and radius $v_I T$.
\end{lemma}

\begin{proof}
    Without loss of generality, let $\mathbf{x}_I^0$ be the origin.
    Based on Lemma \ref{lemma:1D_reach}, all points on the line segment that join the origin and the point $[v_IT\cos\theta,v_IT\sin\theta]^\top$ can be reached with the control function of the form
    \begin{align}
    \mathbf{u}_{1s}(t) = \begin{cases}
    [v_I\cos\theta,v_I\sin\theta]^\top, \quad &\text{if } 0\leq t < t_s,\\
    [-v_I\cos\theta,-v_I\sin\theta]^\top, \quad &\text{if } t_s \leq t \leq T,
    \end{cases}, \label{eq:switch_func_2D}
    \end{align}
    using an appropriate switching time, for all $\theta \in (-\pi/2,\pi/2]$.
    The points on the boundary of the reachable set, which are of the form $[v_IT\cos\theta,v_IT\sin\theta]^\top$, are reached by choosing the control $\mathbf{u}_I(t) = [v_I\cos\theta,v_I\sin\theta]^\top$, $t \in [0,T]$.
    Therefore, in the two-dimensional case
    \begin{align}
        \mathcal{R}_I(T) = \Big\lbrace\mathbf{x} \in \mathbb{R}^2 : \|\mathbf{x}\|_2 \leq v_IT\Big\rbrace.
    \end{align}
\end{proof}


\section{Dependent Reachable Sets}
\label{sec:main_result} 

The solution to Problem \ref{prob:drs} involves identifying a specialized reachable set of the dependent agent, termed the dependent reachable set.
DRS is defined for the agent that determines its control input using a feedback control function of the form $\mathbf{u}_D(t) = f(\mathbf{x}_D(t),\mathbf{x}_I(t),\mathbf{u}_I(t))$, which is a function of the instantaneous state and the control input of the independent agent.
A formal definition for the DRS is presented below.

\begin{definition}
The dependent reachable set of the dependent agent that employs the feedback strategy $f$, at time $t\geq 0$ is defined as
\begin{align}
    \mathcal{D}_f(\mathbf{x}_D^0,\mathbf{x}_I^0,t) &= \Bigg\lbrace\mathbf{x}\in\mathbb{R}^2 : \exists~\hat{\mathbf{u}}_I \in \mathcal{U}_I,~\mathbf{x} = \mathbf{x}_D^0 + \int_0^t f(\hat{\mathbf{x}}_D(\tau),\hat{\mathbf{x}}_I(\tau),\hat{\mathbf{u}}_I(\tau)) \mathrm{d} \tau \Bigg\rbrace,
\end{align}
where $\hat{\mathbf{x}}_D(\tau) = \mathbf{x}_D^0 + \int_0^\tau f(\hat{\mathbf{x}}_D(\eta),\hat{\mathbf{x}}_I(\eta),\hat{\mathbf{u}}_I(\eta)) \mathrm{d} \eta$ and $\hat{\mathbf{x}}_I(\tau) =\mathbf{x}^0_I + \int_0^\tau \hat{\mathbf{u}}_I(\eta) \mathrm{d} \eta$.
\end{definition}

For the problem described in Section \ref{sec:prob}, the reachable set of the independent agent $\mathcal{R}_I(t)$ is a circular region with its center at $\left(a,0\right)$ and radius $v_It$ (see Lemma \ref{lemma:2D_reach}).
Similarly, the reachable set of the dependent agent $\mathcal{R}_D(t)$ is also a circular region with its center at the origin and radius $v_D t$:
\begin{align}
    \mathcal{R}_D(t) &= \Bigg\lbrace\mathbf{x}\in\mathbb{R}^2 : \exists~\mathbf{u}_D \in \mathcal{U}_D,~\mathbf{x} = \mathbf{x}_D^0 + \int_0^t \mathbf{u}_D(\tau) \mathrm{d} \tau \Bigg\rbrace\nonumber \\
    &= \Big\lbrace\mathbf{x}\in\mathbb{R}^2: \|\mathbf{x}\|_2 \leq v_Dt\Big\rbrace. \label{eq:drs_reach_subset}
\end{align}
The boundary of the dependent agent's reachable set is denoted by $\partial \mathcal{R}_D(t)$.
Note that
\begin{align}
    \mathcal{D}_f(\mathbf{x}_D^0,\mathbf{x}_I^0,t) \subseteq \mathcal{R}_D(t).
\end{align}

In this paper, dependent reachable sets are analyzed in the situation where the dependent agent follows the constant bearing pursuit strategy.
The DRS for the constant bearing pursuit strategy is denoted by $\mathcal{D}(t)$, dropping the dependencies on the initial conditions $\mathbf{x}_D^0$ and $\mathbf{x}_I^0$ for brevity.

\subsection{Theoretical Results}
\label{subsec:Theory}

This subsection presents theoretical results for the characterization of DRS $\mathcal{D}(t)$.
As shown in Figure \ref{fig:main_instances}, two different scenarios are considered to develop proofs that characterize the shape of the DRS at a given instant in time $t \in [0,t_c]$.
In this regard, the line segment joining the points $(a,v_It)$ and $(a,-v_It)$ is denoted as $\mathcal{V}$, which is the set of all points that form the vertical diameter of the circle $\partial \mathcal{R}_I(t)$.
The line segment $\mathcal{V}$ is perpendicular to the initial LOS (the horizontal axis) and is represented by a dashed line in Figure \ref{fig:main_instances}.

The first scenario in Figure \ref{fig:main_instances} depicts the case of $\mathcal{V} \nsubseteq \mathcal{R}_D(t)$. 
In the first scenario, there are two possibilities: 1) the line segment $\mathcal{V}$ is outside the circle $\partial \mathcal{R}_D(t)$ (shown in Figure \ref{fig:main_instances}); or 2) the line segment $\mathcal{V}$ intersects the circle $\partial \mathcal{R}_D(t)$.
The first time instant when $\partial \mathcal{R}_D(t)$ intersects the line segment $\mathcal{V}$ is at $t_1 = a/v_D$, which is when $\mathcal{V}$ is tangential to the circle $\partial \mathcal{R}_D(t)$.
Note that at time $t_1$, the initial position of the independent agent $[a,0]^\top$ lies on the circle $\partial \mathcal{R}_D(t)$.
Subsequently, let $t_2$ be the last time instant when the circle $\partial \mathcal{R}_D(t)$ intersects the line segment $\mathcal{V}$.
Due to the symmetry associated with the problem, which can be observed in Figure \ref{fig:main_instances}, it can be deduced that at time $t_2$, $\partial \mathcal{R}_D(t)$ contains the end points of the line segment $\mathcal{V}$, $(a,v_It)$ and $(a,-v_It)$.
As a result, we have
\begin{align}
    v_D^2 t_2^2 &= a^2 + v_I^2 t_2^2 \nonumber \\
    \implies t_2 &= \frac{a}{\sqrt{v_D^2 - v_I^2}}.
\end{align}
Therefore, the first scenario, $\mathcal{V} \nsubseteq \mathcal{R}_D(t)$, corresponds to instances that occur in the time interval $0\leq t < t_2$.

\begin{figure}[htb!]
\centering
\includegraphics[width = 0.6\textwidth]{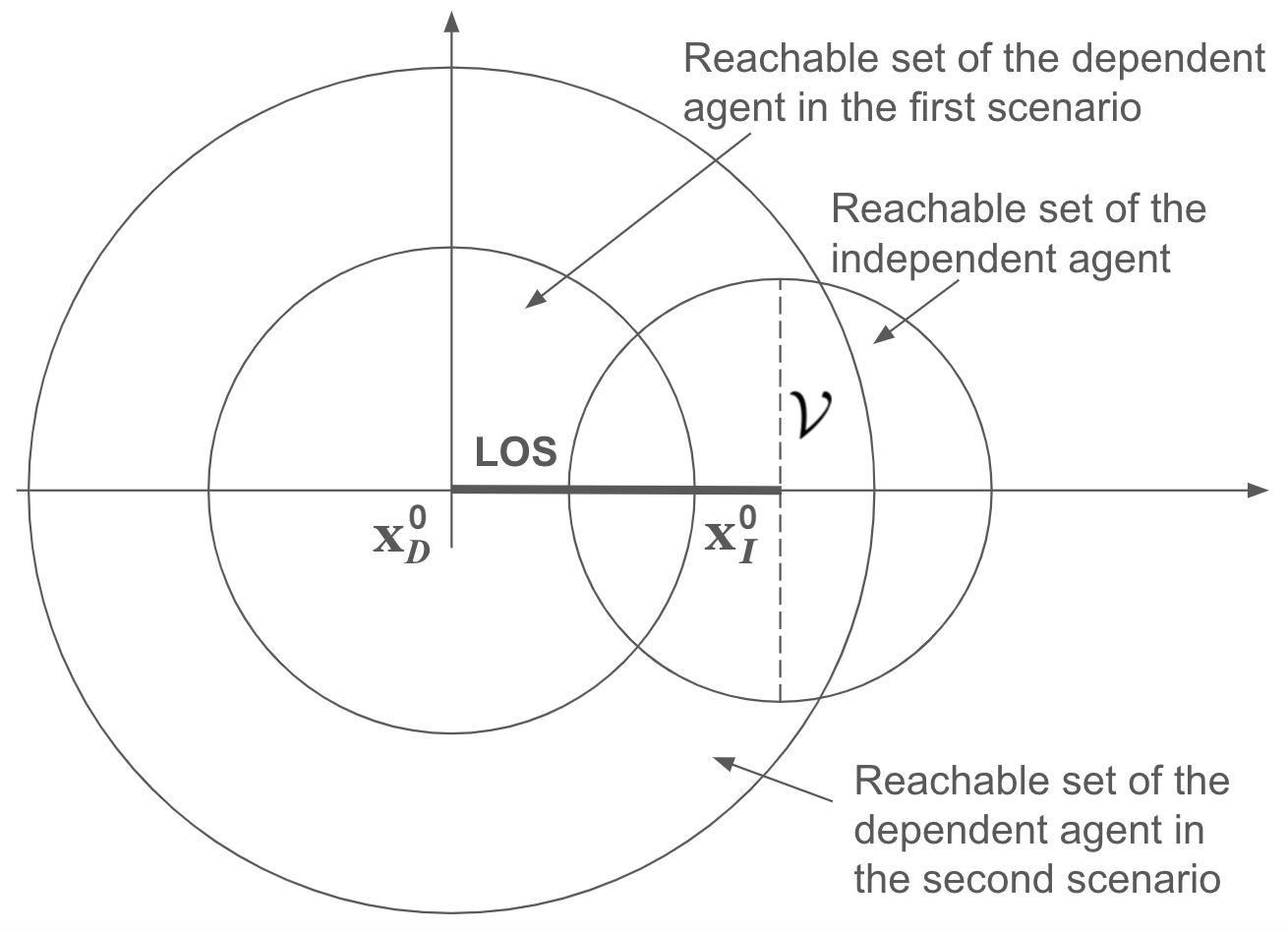}
\caption{Schematics of the two distinct scenarios that are considered to characterize the dependent reachable set for the constant bearing pursuit strategy}
\label{fig:main_instances}
\end{figure}

The second scenario in Figure \ref{fig:main_instances} depicts the case where the circle $\partial \mathcal{R}_D(t)$ fully encompasses the line segment $\mathcal{V}$ i.e. $\mathcal{V} \subseteq \mathcal{R}_D(t)$. 
The instances corresponding to the second scenario occur in the time interval $t_2 \leq t \leq t_c$.
The following lemmas help characterize the DRS for the constant bearing pursuit strategy.

\begin{lemma}\label{lemma:z2_limit}
    For $0 \leq t \leq t_c$, $\mathcal{D}(t) \subseteq \Big\lbrace\mathbf{x} = [x,y]^\top\in\mathbb{R}^2:\mathrm{abs}(y) \leq v_It\Big\rbrace$, where $\mathrm{abs}(.)$ denotes the absolute value.
\end{lemma}
\begin{proof}
    For the constant bearing pursuit strategy, the LOS, which is initially aligned along the horizontal axis, does not rotate. 
    As a result, $y_D(t) = y_I(t)$ for all time $t \in [0,t_c]$, and the vertical coordinate of the dependent agent is bounded by $\pm v_It$.
    Therefore,
    \begin{align}
        \mathcal{D}(t) \subseteq \Big\lbrace\mathbf{x} = [x,y]^\top\in\mathbb{R}^2:\mathrm{abs}(y) \leq v_It\Big\rbrace. \label{eq:z2_limit}
    \end{align}
\end{proof}

\begin{lemma}\label{lemma:z1_limit}
    For $0 \leq t \leq t_c$, $\mathcal{D}(t) \subseteq \Big\lbrace\mathbf{x} = [x,y]^\top\in\mathbb{R}^2:x \geq t\sqrt{v_D^2 - v_I^2}\Big\rbrace$.
\end{lemma}
\begin{proof}
    For this proof, the minimum horizontal speed achievable by the dependent agent with the constant bearing pursuit strategy must be identified. 
    To this end, since the dependent agent's heading ($\psi_D(t)$) is a function of the independent agent's heading ($\psi_I(t)$) per (\ref{eq:const_bear}), the $\psi_I(t)$ corresponding to the minimum horizontal speed can be obtained as
    \begin{align}
        \underset{\psi_I(t) \in (-\pi,\pi]}{\arg\min}~\cos\psi_D(t) = \pm \dfrac{\pi}{2}.
    \end{align}
    Note that $\psi_I(t) = \pm \pi/2$ indicates that the independent agent moves perpendicular to the LOS.
    Consequently, at time $t$, the minimum horizontal speed of the dependent agent is 
    \begin{align}
        \underset{\psi_I(t) \in (-\pi,\pi]}{\min}~v_D\cos\psi_D(t) &= v_D \cos \Big( \sin^{-1} \left(\dfrac{v_I}{v_D}\right)\Big) \nonumber \\
        &= \sqrt{v_D^2- v_I^2}.
    \end{align}
    As a result, the minimum horizontal distance that the dependent agent can travel within time $t$ is $t\sqrt{v_D^2- v_I^2}$.
    Therefore,
    \begin{align}
        \mathcal{D}(t) \subseteq \Big\lbrace\mathbf{x} = [x,y]^\top\in\mathbb{R}^2:x \geq t\sqrt{v_D^2 - v_I^2}\Big\rbrace. \label{eq:z1_limit}
    \end{align}
\end{proof}

From Lemmas \ref{lemma:z2_limit} and \ref{lemma:z1_limit}, it can be deduced that $\mathcal{D}(t)$ will lie in the intersection of the regions given in (\ref{eq:drs_reach_subset}), (\ref{eq:z2_limit}), and (\ref{eq:z1_limit}).
The intersection region for the first scenario ($0\leq t\leq t_2$) is shown in Figure \ref{fig:DRS_shape}.
The boundaries of these three regions intersect at the points $P_1$ and $P_2$, as shown in Figure \ref{fig:DRS_shape}.
The coordinates of the points $P_1$ and $P_2$ can be obtained as $\left(t\sqrt{v_D^2 - v_I^2},v_It\right)$ and $\left(t\sqrt{v_D^2 - v_I^2},-v_It\right)$, respectively.
From the aforementioned coordinates, it can be seen that the points $P_1$ and $P_2$ lie on the circle $\partial\mathcal{R}_D(t)$.
Note that at time $t_2$, the line segment joining the points $P_1$ and $P_2$ coincides with $\mathcal{V}$, the vertical diameter of the circle $\partial \mathcal{R}_I(t)$.
In Figure \ref{fig:DRS_shape}, $P_x = [v_D t, 0]^\top$.
The discussion of the Apollonius circle (shown in Figure \ref{fig:DRS_shape}) and its link to the DRS is included after Theorem \ref{thm:main}.
The following theorem establishes that the DRS for the first scenario is the intersection region itself.

\begin{figure}[htb!]
\centering
\includegraphics[width = 0.65\textwidth]{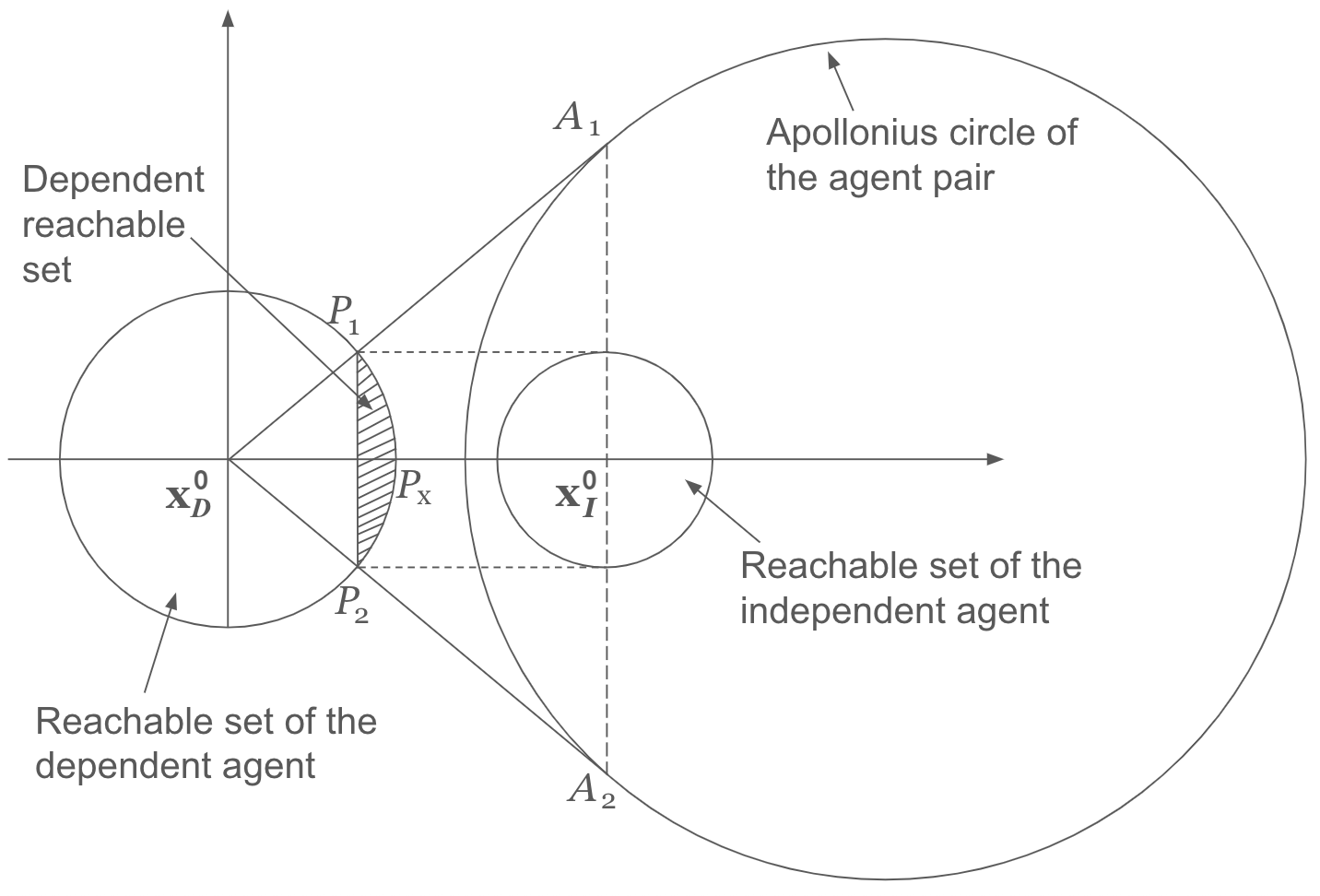}
\caption{Characterization of the DRS for $0\leq t \leq t_2$}
\label{fig:DRS_shape}
\end{figure}

\begin{theorem}\label{thm:main}
For $0 \leq t \leq t_2$, 
\begin{align}
    \mathcal{D}(t) = \Big\lbrace\mathbf{x}= [x,y]^\top\in\mathbb{R}^2: \|\mathbf{x}\|_2 \leq v_Dt \text{ and } x \geq t\sqrt{v_D^2 - v_I^2}\Big\rbrace.\label{eq:drs_thm}
\end{align}
\end{theorem}
\begin{proof}
The expressions in (\ref{eq:drs_reach_subset}) and (\ref{eq:z1_limit}) establish that the sets $\|\mathbf{x}\|_2 \leq v_Dt$ and $x \geq t\sqrt{v_D^2 - v_I^2}$ bound the DRS for all time $t\in [0,t_c]$.
The proof involves first showing that for every point $\mathbf{x}$ from the intersection of the aforementioned sets, there exists a control function $\mathbf{u}_I(.) \in \mathcal{U}_I$ of the independent agent which drives the dependent agent following the constant bearing pursuit strategy to the point $\mathbf{x}$.
In this regard, for time $t \in (0, t_2]$ and $\theta \in \left[-\dfrac{\pi}{2}, \dfrac{\pi}{2}\right]$, consider the candidate switching control function of the type 
\begin{align}
\mathbf{u}_{2s}(\tau) = \begin{cases}
[v_I\cos\theta,v_I\sin\theta]^\top, \quad &\text{if } 0\leq \tau < t_s,\\
[0,v_I]^\top, \quad &\text{if } t_s \leq \tau < \dfrac{t+t_s}{2},\\
[0,-v_I]^\top, \quad &\text{if } \dfrac{t+t_s}{2} \leq \tau \leq t,
\end{cases}, \label{eq:switch_func_2D_2}
\end{align}
for the independent agent.
Since the independent agent starts at $(a,0)$, for the control function in (\ref{eq:switch_func_2D_2}), the independent agent reaches the point $\left(a+t_s v_I\cos\theta, t_s v_I\sin\theta\right)$ at time $t_s$, then travels vertically up until time $\dfrac{t+t_s}{2}$, and finally returns to the point $\left(a+t_s v_I\cos\theta, t_s v_I\sin\theta\right)$ at time $t$.
Figure \ref{fig:thm_2s_proof} indicates the resulting path (lines with arrows that depict the direction of motion) of the independent agent for the control function in (\ref{eq:switch_func_2D_2}). 

For the control function in (\ref{eq:switch_func_2D_2}), $\psi_D(t)$ for time $t \in [0,t_s)$ can be obtained as $\sin^{-1}\left(\dfrac{v_I \sin\theta}{v_D}\right)$ per (\ref{eq:const_bear}).
As a result, the dependent agent reaches the point $\left(t_s\sqrt{v_D^2-v_I^2\sin^2\theta},t_s v_I\sin\theta\right)$ at time $t_s$.
Note that the dependent agent starts from the origin.
As discussed in the proof of Lemma \ref{lemma:z2_limit}, the horizontal speed of the dependent agent for the time when the independent agent moves perpendicular to the LOS (horizontal axis) is $\sqrt{v_D^2 - v_I^2}$.
Therefore, at time $t$, the dependent agent reaches the horizontal coordinate 
\begin{align}
    t_s\sqrt{v_D^2 - v_I^2\sin^2\theta} + (t-t_s)\sqrt{v_D^2 - v_I^2}.\label{eq:xd_second_thm}
\end{align} 
Note that the corresponding vertical coordinate is $t_s v_I\sin\theta$, which is the same as that of the independent agent.

Given $y=t_s v_I\sin\theta$, the dependent agent's horizontal coordinate in (\ref{eq:xd_second_thm}) can be rewritten as
\begin{align}
\sqrt{t^2_sv_D^2 - y^2} + (t-t_s)\sqrt{v_D^2 - v_I^2}. \label{eq:xd_second_thm_2}
\end{align}
From (\ref{eq:xd_second_thm_2}), it can be observed that all the points from the set in (\ref{eq:drs_thm}) are spanned using $t_s \in [y/v_I, t]$, given $y \in [-v_It, v_It]$.
Each point from the set in (\ref{eq:drs_thm}) corresponds to the point reached by the independent agent using the control function in (\ref{eq:switch_func_2D_2}) for $\pi/2 \leq \theta \leq \theta_l$ (when $y \geq 0$) or $\theta_l \leq \theta \leq -\pi/2$ (when $y < 0$), where $\theta_l = \sin^{-1}(y/v_It)$.
   
\begin{figure}[htb!]
\centering
\includegraphics[width = 0.6\textwidth]{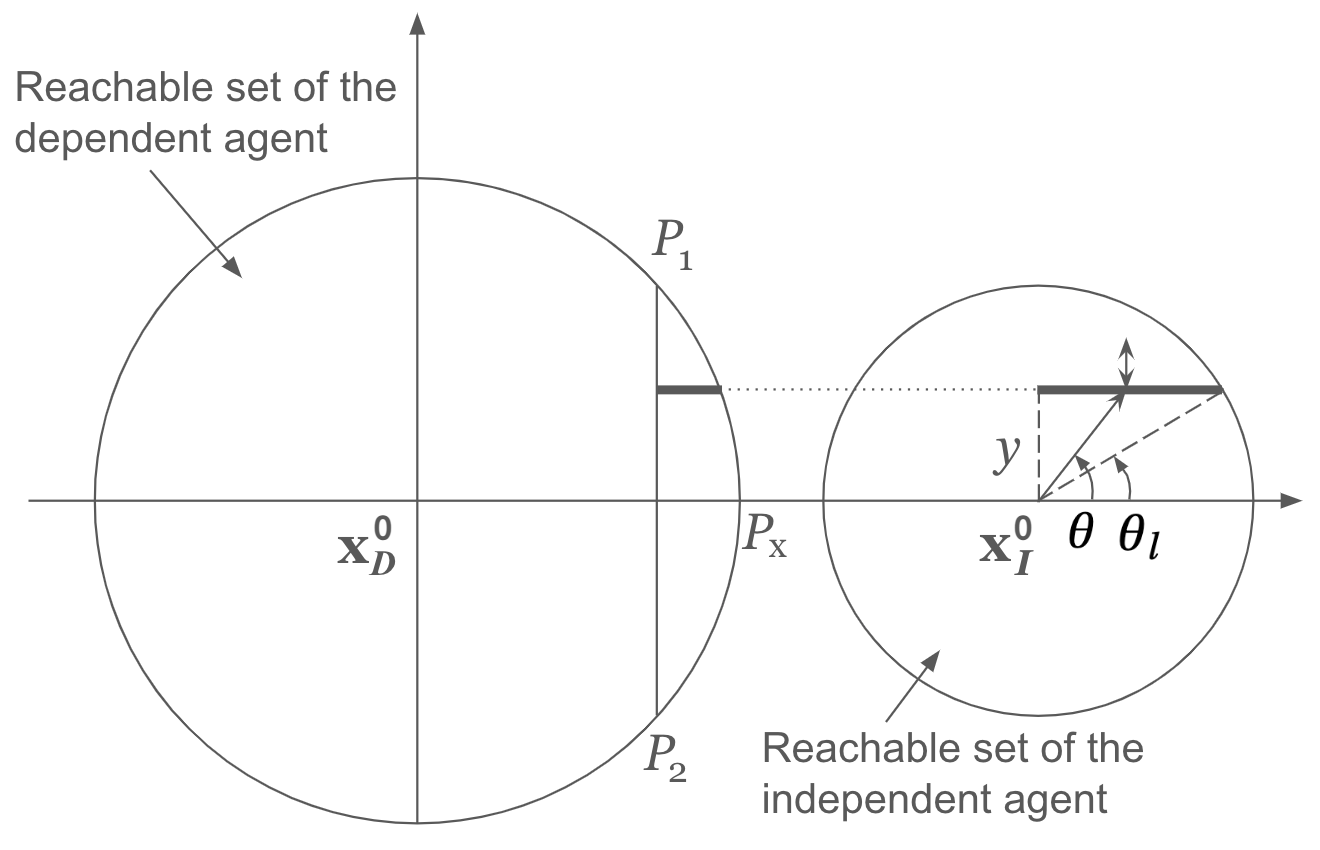}
\caption{Independent agent's path for the control function in (\ref{eq:switch_func_2D_2}), and the resulting points of both agents for a given vertical coordinate $y$. }
\label{fig:thm_2s_proof}
\end{figure}

Next, we need to show that all the resulting trajectories of the dependent agent corresponding to the control function in (\ref{eq:switch_func_2D_2}), and $\theta \in \left[-\dfrac{\pi}{2}, \dfrac{\pi}{2}\right]$ will not result in capture for $0 \leq t < t_2$.
For a capture to occur at time $t$ with the independent agent's control function in (\ref{eq:switch_func_2D_2}), both agents should reach the same horizontal coordinate at time $t$.
In this regard, for the capture at $t$,
\begin{align}
&t_s\sqrt{v_D^2 - v_I^2\sin^2\theta} + (t-t_s)\sqrt{v_D^2 - v_I^2} = a + v_It_s\cos\theta, \nonumber\\
\implies &t = \dfrac{a}{\sqrt{v_D^2 - v_I^2}} + \dfrac{t_s}{\sqrt{v_D^2 - v_I^2}}\Bigg[ v_I\cos\theta + \sqrt{v_D^2 - v_I^2} - \sqrt{v_D^2 - v_I^2\sin^2\theta}\Bigg]\nonumber\\
&~~= \dfrac{a}{\sqrt{v_D^2 - v_I^2}} + \dfrac{t_s}{\sqrt{v_D^2 - v_I^2}}\Bigg[ v_I\cos\theta + \sqrt{v_D^2 - v_I^2} - \sqrt{v_D^2 - v_I^2 + v_I^2\cos^2\theta}\Bigg].\label{eq:thm_final}
\end{align}
Since $t_s \geq 0$, $v_D > v_I$ and $\theta \in \left[-\dfrac{\pi}{2}, \dfrac{\pi}{2}\right]$, (\ref{eq:thm_final}) yields the capture time $t \geq a/\left(\sqrt{v_D^2 - v_I^2}\right) = t_2$. 
Consequently, the dependent agent cannot capture the independent agent, which executes the control function in (\ref{eq:switch_func_2D_2}), using the constant bearing pursuit strategy before time $t_2$.
Therefore, the dependent agent can reach all the points from the set in (\ref{eq:drs_thm}) via active pursuit trajectories.
Hence proved.
\end{proof}

As discussed in \cite{makkapati2019optimal}, the constant bearing pursuit strategy and the Apollonius circle have close connections.
The Apollonius circle is defined for a pursuer-evader pair as the set of capture points when the pursuer follows the constant bearing pursuit strategy, and the evader chooses a constant heading angle. 
The geometric properties of the Apollonius circle can be found in a prior study \cite{ramana2017pursuit}. 
Given the initial positions of the dependent and independent agents as $(0,0)$ and $(a,0)$, the center and radius of the corresponding Apollonius circle $\mathcal{A}$ can be obtained as  $\left(av_D^2/(v^2_D - v_I^2),0\right)$ and $av_Dv_I/v_D^2 - v_I^2$, respectively, as shown in Figure \ref{fig:DRS_shape}.

\begin{lemma}
    For $0 \leq t \leq t_2$, the points $P_1$ and $P_2$ lie on either tangent lines of the Apollonius circle $\mathcal{A}$ emanating from the initial position of the dependent agent.
\end{lemma} 
\begin{proof}
The line segment joining the two tangent points $A_1$ and $A_2$, in Figure \ref{fig:DRS_shape}, was shown to pass through the initial position of the independent agent (see Lemma 3.2 in Ref. \cite{ramana2017pursuit}).
Furthermore, the line joining $A_1$ and $A_2$ subtends an angle of $2\sin^{-1}(v_I/v_D)$ at the initial position of the dependent agent, and the line is perpendicular to the initial LOS (see Lemma 3.1 in Ref. \cite{ramana2017pursuit}).
Given the coordinates of $P_1$ and $P_2$ as $\left(t\sqrt{v_D^2 - v_I^2},\pm v_It\right)$, the line joining the origin and $P_1$ (or $P_2$) is angled at $\sin^{-1}(v_I/v_D)$ from the horizontal axis.
Consequently, the line joining the points $P_1$ and $P_2$ subtends an angle of $2\sin^{-1}(v_I/v_D)$ at the initial position of the dependent agent, and the line is perpendicular to the LOS.
Therefore, the points $A_1$, $P_1$, and the initial position of the dependent agent are collinear.
Similarly, the points $A_2$, $P_2$, and the initial position of the dependent agent are collinear.
Hence proved.
\end{proof}

It can be noted that at time $t_2$, the initial position of the independent agent $\mathbf{x}^0_I=(a,0)$ lies on the line segment $\overline{P_1P_2}$ \cite{ramana2017pursuit}.
Therefore, at time $t_2$, $\overline{P_1P_2}$ coincides with the line segment $\overline{A_1A_2}$ at time $t_2$.

\begin{corollary}
For $t_2 < t \leq t_c$, $\mathcal{D}(t) \subseteq \Big\lbrace\mathbf{x}= [x,y]^\top\in\mathbb{R}^2: \|\mathbf{x}\|_2 \leq v_Dt \text{ and } x \geq t\sqrt{v_D^2 - v_I^2}\Big\rbrace$.   
\end{corollary}
\begin{proof}
    The proof directly follows from Lemmas \ref{lemma:z2_limit}, \ref{lemma:z1_limit}.
\end{proof}

The above corollary provides a bound for the DRS in the second scenario ($t_2 < t \leq t_c$).
The bound is characterized by the points $P_1$ and $P_2$, and is depicted in Figure \ref{fig:scene2_drs}.
It is important to note that the above corollary, while providing a tight bound for the DRS, does not define the boundary of the DRS. 
The mathematical proof to establish the boundary of the DRS for $t_2 < t \leq t_c$ is elusive.
While a formal mathematical proof remains an open challenge, we present compelling evidence using simulation results (in the following subsection) that supports our hypothesis below.

\begin{hypothesis}
For $t_2 < t \leq t_c$,
\begin{align}
    \mathcal{D}(t) = \Bigg\lbrace\mathbf{x}= [x,y]^\top\in\mathbb{R}^2: \|\mathbf{x}\|_2 \leq v_Dt \text{ and } x \geq \dfrac{a^2 + v_D^2t^2 - v_I^2t^2}{2a}\Bigg\rbrace.\label{eq:drs_hypo}
\end{align}
\end{hypothesis}

The vertical line $x = (a^2 + v_D^2t^2 - v_I^2t^2)/2a$ contains the points $Q_1$ and $Q_2$, which are the intersection points for the boundaries of the reachable sets of both agents ($\partial \mathcal{R}_I$ and $\partial \mathcal{R}_D$), as shown in Figure \ref{fig:scene2_drs}.
In Section \ref{subsec:simres}, it is empirically shown that the DRS for time $t_2 < t \leq t_c$ is the shaded region in Figure \ref{fig:scene2_drs}, which is characterized by the minor segment $Q_1 P_x Q_2$.

\begin{figure}[htb!]
\centering
\includegraphics[width = 0.6\textwidth]{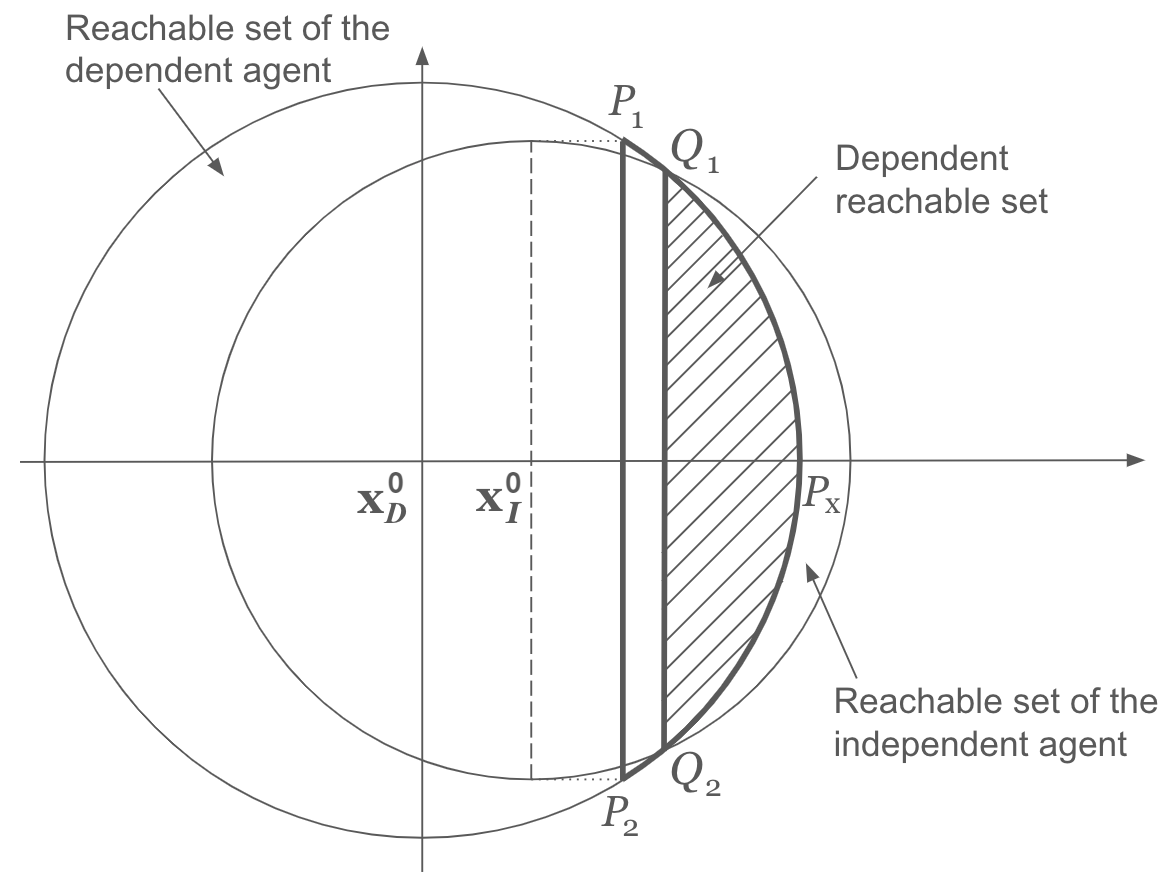}
\caption{Characterization of the DRS for $t_2 < t \leq t_c$}
\label{fig:scene2_drs}
\end{figure}

\subsection{The Limiting Case of $v_D = v_I$}\label{subsec:limit_case}

In the case of $v_D = v_I$, at time $t$, the points $P_1$ and $P_2$ are obtained as $(0,v_I t)$ and $(0, -v_I t)$, respectively.
Therefore, the points $P_1$, $P_2$, and the initial point of the dependent agent are collinear.
The line joining the three points forms the vertical diameter of the circle $\partial \mathcal{R}_D(t)$ (boundary of the dependent agent's reachable set), which can be visualized in Figure \ref{fig:equal_drs}

\begin{figure}[htb!]
\centering
\includegraphics[width = 0.6\textwidth]{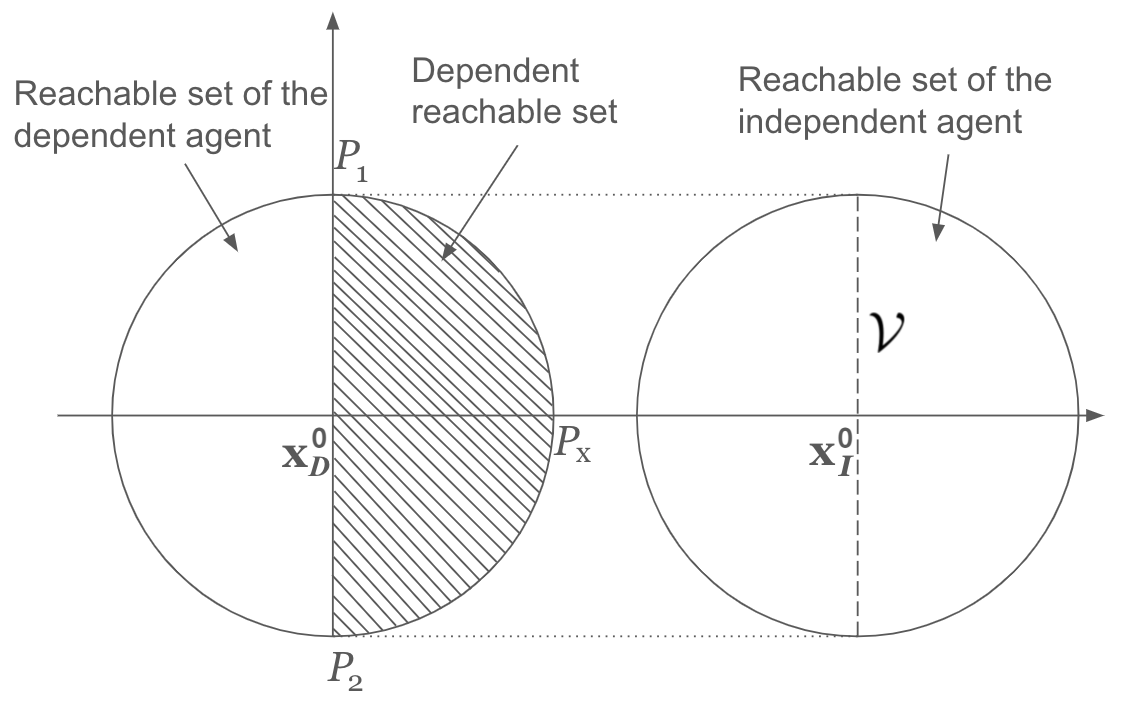}
\caption{DRS for the limiting case $v_D = v_I$}
\label{fig:equal_drs}
\end{figure}

Note that as $v_D \rightarrow v_I$, $t_2 \rightarrow \infty$.
Consequently, only the first scenario discussed in Section \ref{subsec:Theory} is relevant for the limiting case, and the second scenario is irrelevant.
Furthermore, it can be observed that the conclusion drawn in the proof of Theorem \ref{thm:main} is valid for the case of $v_D=v_I$.
Consequently, the DRS is given by
\begin{align}
    \mathcal{D}(t) = \Big\lbrace\mathbf{x}= [x,y]^\top\in\mathbb{R}^2: \|\mathbf{x}\|_2 \leq v_Dt \text{ and } x \geq 0\Big\rbrace.
\end{align}
The arc $\widearc{P_1P_2}$ that contains the point $P_x$, which is part of the boundary of the DRS, forms a semicircle.
The DRS for this case can be visualized in Figure \ref{fig:equal_drs}.

\subsection{Simulation Results}
\label{subsec:simres}
This subsection presents simulation results based on discrete-time point-cloud propagation to observe the evolution of the dependent reachable sets over time $0 \leq t \leq t_c$.
While HJ reachability-based numerical methods provide computational efficiency in estimating reachable sets, the provision to track active trajectories is unavailable. 
Point-cloud propagation is chosen to accurately estimate the DRS by visualizing the active trajectories and removing those trajectories that result in capture.
The initial positions of both agents are $x_D^0(0,0)$ and $x_I^0(1,0)$ and their speeds are set to $v_D = 1$ and $v_I = 0.5$ respectively. 
Simulations are carried out with the time-step $\Delta t = 0.2$ until all the independent agents are captured.
To generate a point cloud, we consider that at a given time instant, an independent agent chooses eighteen equally-spaced heading angles from $(-\pi,\pi]$.
Consequently, at the $n^{th}$ time step, $18^{n}$ points are generated for the independent agent. 
At any time instant, every independent agent point has an associated dependent agent point obtained via the constant bearing pursuit strategy.
For the simulation parameters mentioned above, $t_1 = 1.0$, $t_2 \approx 1.155$ and $t_c = 2.0 s$.

\begin{figure*}
    \subfigure[$t = 0.6  \leq t_1$]{\includegraphics[width = 0.5\textwidth]{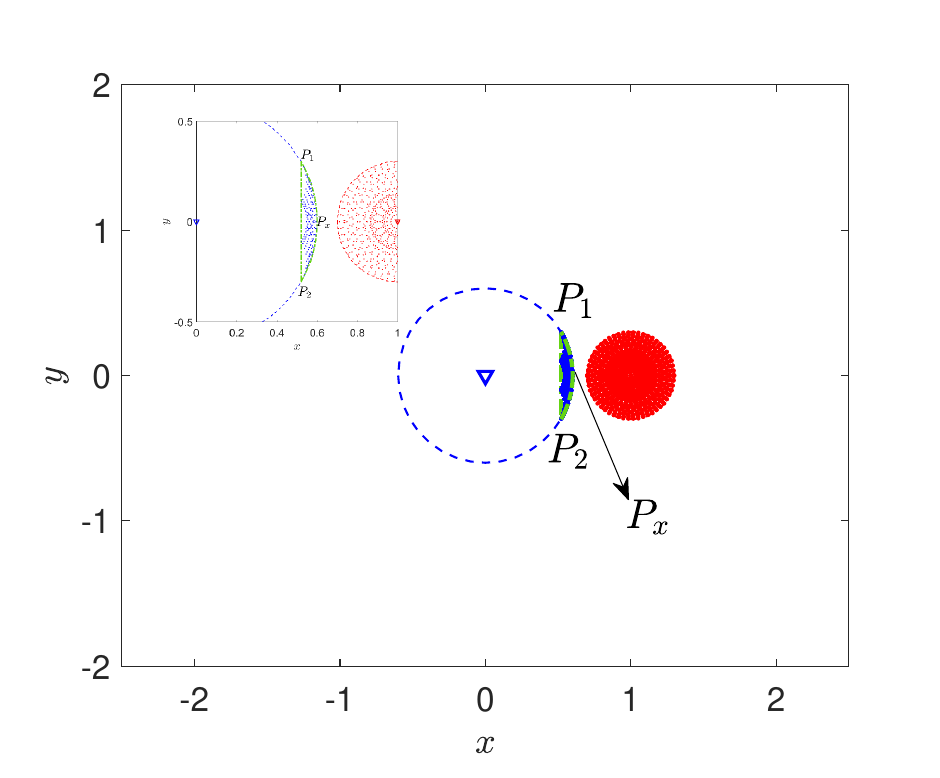}\label{fig:DRS_sim_a}} 
    \subfigure[$t = 1.0  = t_1$]{\includegraphics[width = 0.5\textwidth]{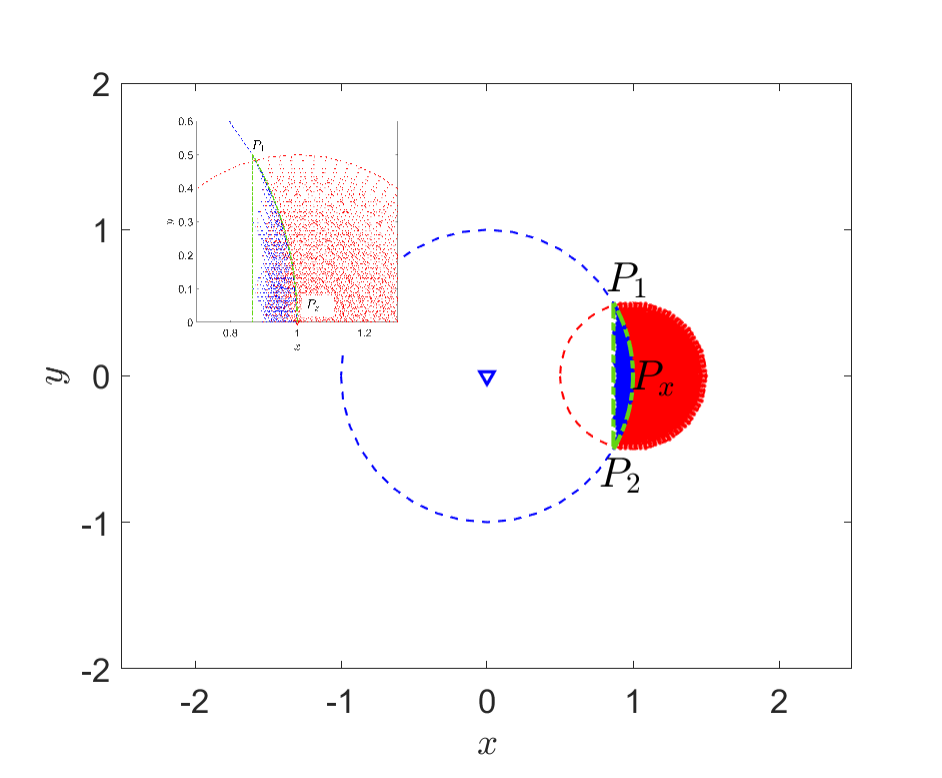}\label{fig:DRS_sim_b}}
    \subfigure[$t = 1.2 \approx t_2$]{\includegraphics[width = 0.5\textwidth]{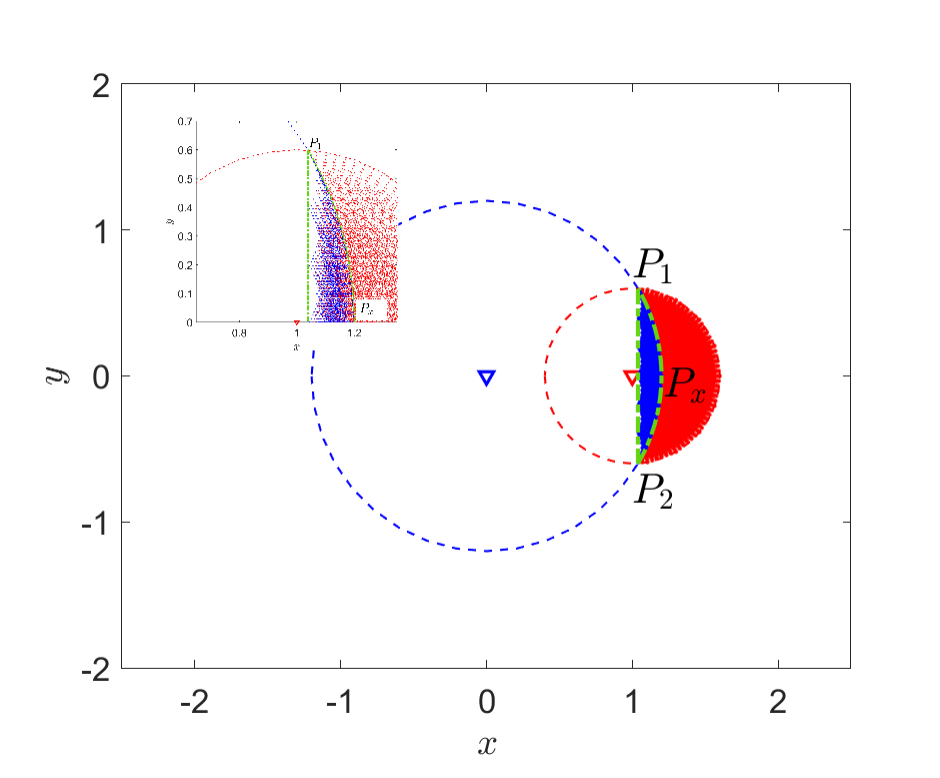}\label{fig:DRS_sim_c}} 
    \subfigure[$t_2 < t = 1.4 < t_c$]{\includegraphics[width = 0.5\textwidth]{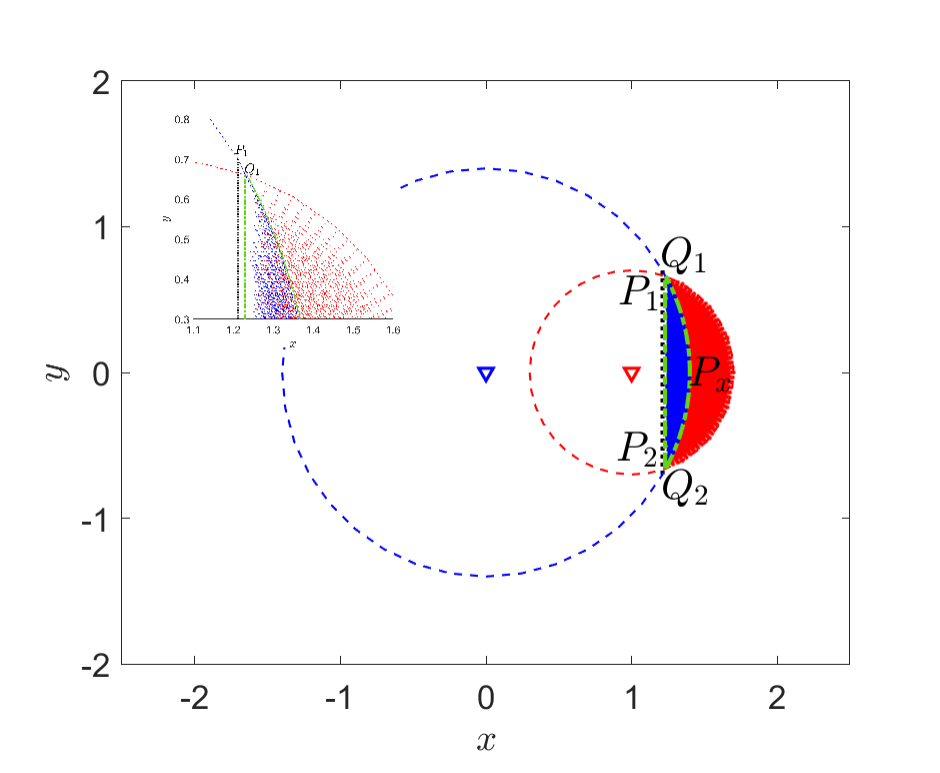}\label{fig:DRS_sim_d}}
    \subfigure[$t_ 2 < t = 1.8 < t_c$]{\includegraphics[width = 0.5\textwidth]{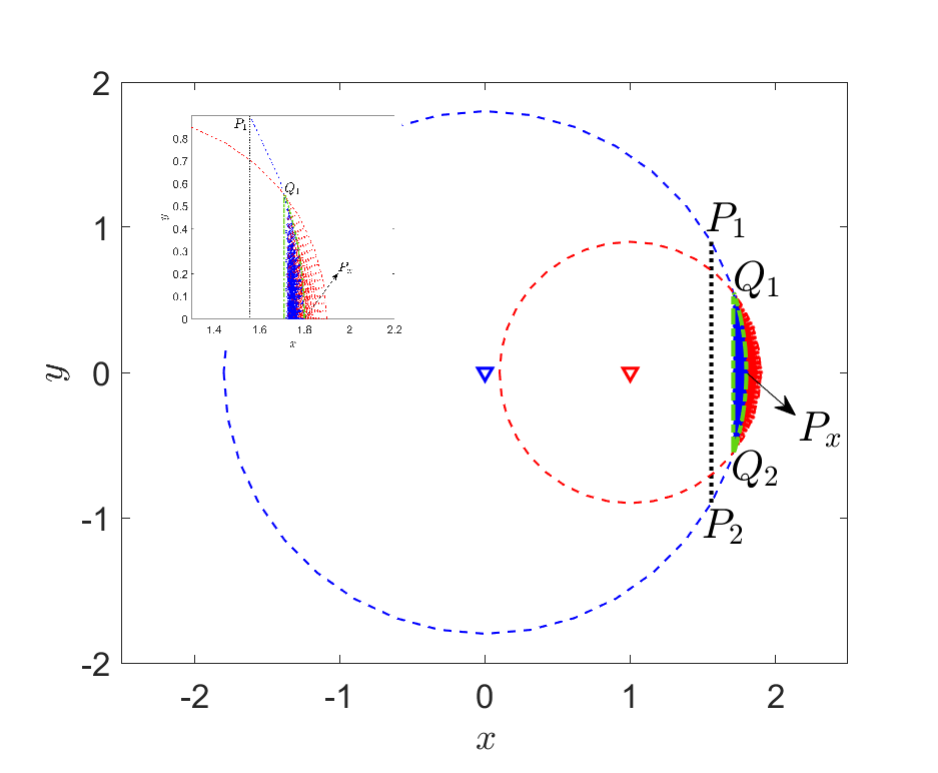}\label{fig:DRS_sim_e}}
    \subfigure[$t = 2.0 = t_c$]{\includegraphics[width = 0.5\textwidth]{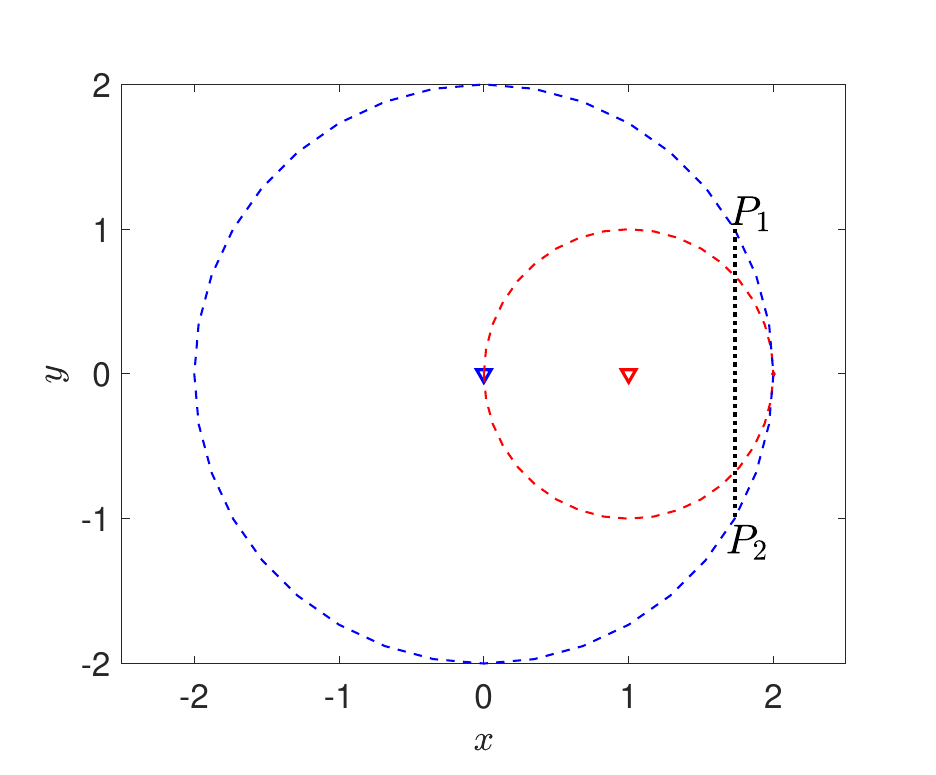}\label{fig:DRS_sim_f}}
    \caption{Point cloud-based simulation results depicting the evolution of DRS for $0 < t\leq t_c$}
    \label{fig:DRS_sim}
\end{figure*}

Figure \ref{fig:DRS_sim} presents the results obtained from the point cloud-based simulations. 
The blue color corresponds to the dependent agent, and the red color corresponds to the independent agent. 
The dashed circles indicate the boundaries of the agents' reachable sets. 
The green line indicates the boundary of the DRS.
The inverted triangles represent the initial positions of the two agents. 
The points in the red point cloud represent uncaptured independent agents, and
the points in the blue point cloud represent the active dependent agents that are yet to capture their corresponding independent agents (refer to Section \ref{sec:prob}). 
As the simulation progresses, when a dependent agent captures an independent agent, that particular pair is eliminated from their respective point clouds. 
Figures \ref{fig:DRS_sim_a} and \ref{fig:DRS_sim_b} represent two instances of the scenario $0 < t \leq t_1$, and Figure \ref{fig:DRS_sim_c} represents the case where $t \approx t_2$.
Figures \ref{fig:DRS_sim_d}-\ref{fig:DRS_sim_f} depict three instances of $t_2 < t \leq t_c$. 

As can be seen in the Figs. \ref{fig:DRS_sim_a} through \ref{fig:DRS_sim_f}, the line segment (chord, in geometric terms) $\overline{P_1P_2}$ bounds the DRS, which also forms part of the DRS boundary for $t \leq t_2$ (Figs. \ref{fig:DRS_sim_a} - \ref{fig:DRS_sim_c}).
Similarly, the minor arc $\widearc{P_1P_2}$ always bounds the DRS on the opposite side of the vertical chord $\overline{P_1P_2}$, and the minor segment $P_1P_xP_2$ represents the DRS for $t \leq t_2$. 
The length of $\overline{P_1P_2}$ is the same as the diameter of the reachable set of the independent agent. 
The active independent agents initially cover the entire independent agent's reachable set, as indicated in Figure \ref{fig:DRS_sim_a}. 
When the two reachable sets begin to intersect and overlap, some independent agents are captured, resulting in partial occupation of the independent agent's reachable set (see Figs. \ref{fig:DRS_sim_b}-\ref{fig:DRS_sim_e}). 

In Figs. \ref{fig:DRS_sim_d} and \ref{fig:DRS_sim_e}, it can be observed that the minor segment $P_1P_xP_2$ contains a region that is outside of the independent agent's reachable set $\mathcal{R}_I(t)$. 
Therefore, the points in the minor segment $P_1P_xP_2$ that are not part of the independent agent's reachable set cannot be part of the DRS.
Furthermore, there are no blue points in the region between the line segments $\overline{P_1P_2}$ and $\overline{Q_1Q_2}$.
It can be observed that the chord $\overline{Q_1Q_2}$ is part of the boundary of the blue point cloud when $t_2 < t \leq t_c$, and the minor segment $Q_1P_xQ_2$ contains all the blue points when $t_2 < t \leq t_c$. 
A more enhanced image of the interaction between the reachable sets is shown in the upper left corner of each figure.
Simulations have been conducted for different parameters ($v_D$, $v_I$, etc.), which have not been included for brevity, and the results provided the same observations.
Thus, it could be empirically stated that the DRS, initially represented by the segment $P_1P_xP_2$, reaches a maximum area at $t = t_2$ and subsequently shrinks to the segment $Q_1P_xQ_2$, which is contained within the segment $P_1P_xP_2$. 
Figure \ref{fig:DRS_sim_f} ($t = t_c$) presents the case where the dependent agent's reachable set fully encompasses the independent agent's reachable set, and all the points from both red and blue point clouds ended in capture.


\section{The Optimization Problem}
\label{sec:optimize}

In this section, an optimization problem is formulated to gain deeper insight into the geometry of the DRS.
A simple approach to determining the DRS at a time $t \geq 0$, $\mathcal{D}(t)$, consists of two steps. 
First, for each possible state $\mathbf{x}$ that the independent agent can reach at time $t$, compute the corresponding set of states $\mathcal{S}(\mathbf{x},t)$ that the dependent agent can reach with the constant bearing pursuit strategy. 
The complete DRS is then obtained as the union of all such sets of states of the dependent agent, i.e. 
\begin{align}
    \mathcal{D}(t) = \underset{\mathbf{x}\in\mathcal{R}_I(t)}{\bigcup}\mathcal{S}(\mathbf{x},t).
\end{align}
As stated earlier in Section \ref{sec:intro}, it is important to note that the independent agent may reach a particular state at time $t$ via multiple distinct trajectories. 
In such cases, the dependent agent can potentially reach multiple different states from the same initial condition with the constant bearing pursuit strategy.

Now, consider a point $\mathbf{x}_I(t) = \mathbf{x} = [x,y]^\top \in \mathcal{R}_I(t)$, where the independent agent reaches at time $0 \leq t \leq t_c$. 
Note that if $\|\mathbf{x}\|_2 < v_It$, there are infinitely many trajectories that the independent agent can take to reach the point $\mathbf{x}$.
The goal is to find all the points that the dependent agent, which follows the constant bearing pursuit strategy, can reach for the given $\mathbf{x}$.

As formulated in Section \ref{sec:prob}, since the dependent agent follows the constant bearing pursuit strategy and given the initial conditions of both agents (origin and $(a,0)$), the vertical coordinates of both agents are the same for all time $t \geq 0$.
Therefore, the vertical coordinate of the dependent agent at time $t$ is the same as the independent agent's vertical coordinate, which is $y$.
Subsequently, we intend to find the minimum and the maximum horizontal coordinates that the dependent agent can reach at the time instant $t$. 
Since $\dot{y}_D(\tau) = \dot{y}_I(\tau)$ for the constant bearing pursuit strategy, $\dot{x}_D(\tau) = \sqrt{v_D^2 - \dot{y}_I^2(\tau)}$.
In this regard, consider the problem of obtaining the extrema for the functional
\begin{align}
    \int_0^t\sqrt{v_D^2 - \dot{y}_I^2(\tau)}~ \mathrm{d} \tau, \label{eq:optim_cost}
\end{align}
for the constraints
\begin{align}
    \dot{x}_I^2(\tau)+\dot{y}_I^2(\tau) = v_I^2,\label{eq:cv_constraint}
\end{align}
\begin{align}
    \mathbf{x}_I(0) = [a,0]^\top, \quad \mathbf{x}_I(t) = \mathbf{x} = [x,y]^\top. \label{eq:optim_bounds}
\end{align}
The constraint in (\ref{eq:cv_constraint}) can be absorbed into the optimization by considering the Lagrangian
\begin{align}
    L = \sqrt{v_D^2 - \dot{y}_I^2(\tau)} + \lambda(\tau)\left[\dot{x}_I^2(\tau)+\dot{y}_I^2(\tau) - v_I^2\right],
\end{align}
where $\lambda(\tau)$ is the auxiliary variable.
The first-order necessary conditions for the extremum points can be obtained using the corresponding Euler-Lagrange equations given below.
\begin{align}
    &\dfrac{\partial L}{\partial x_I} - \dfrac{\mathrm{d}}{\mathrm{d}\tau}\left(\dfrac{\partial L}{\partial \dot{x}_I}\right) = -\dfrac{\mathrm{d}}{\mathrm{d}\tau}\Bigg(2\dot{x}_I(\tau)\lambda(\tau)\Bigg) = 0,\label{eq:EL_1}\\
    &\dfrac{\partial L}{\partial y_I} - \dfrac{\mathrm{d}}{\mathrm{d}\tau}\left(\dfrac{\partial L}{\partial \dot{y}_I}\right) = -\dfrac{\mathrm{d}}{\mathrm{d}\tau}\Bigg(-\dfrac{\dot{y}_I(\tau)}{\sqrt{v_D^2 - \dot{y}_I^2(\tau)}}+2\lambda(\tau)\dot{y}_I(\tau)\Bigg) = 0.\label{eq:EL_2}
\end{align}
Consequently,
\begin{align}
    \dot{x}_I(\tau)\lambda(\tau) = c_1, \label{eq:EL1}
\end{align}
\begin{align}
    \dot{y}_I(\tau)\left[1/(\sqrt{v_D^2 - \dot{y}_I^2(\tau)})-2\lambda(\tau)\right] = c_2, \label{eq:EL2}
\end{align}
for $0\leq \tau \leq t$, where $c_1$ and $c_2$ are constants.

Equations (\ref{eq:EL_1}) and (\ref{eq:EL_2}) cannot be solved further to obtain analytical expressions of the extrema. 
Therefore, simulation-based empirical studies were conducted to understand the nature of the extrema in the case where the independent agent follows a single-switching control function, which is the simplest candidate function, to reach the point $\mathbf{x}$.
Note that in such cases, the locus of points at which the independent agent switches between constant control inputs is the ellipse 
\begin{align}
    \mathcal{E} = \Big\lbrace\mathbf{x}_s = [x_s,y_s]^\top:\|\mathbf{x}_s - \mathbf{x}_I(0)\| + \|\mathbf{x}_s - \mathbf{x}\| = v_It\Big\rbrace. \label{eq:switch_locus}
\end{align}
The following hypotheses are proposed on the basis of empirical evidence.

\begin{hypothesis}\label{hypo:optim_1}
    The extremum points for the optimization problem in (\ref{eq:optim_cost})-(\ref{eq:optim_bounds}) correspond to the independent agent's trajectories with a single-switching control function such that $\mathrm{abs}(\dot{x}(\tau))$ and $\mathrm{abs}(\dot{y}(\tau))$ are constants for all time $0\leq \tau \leq t$.
\end{hypothesis}

\begin{hypothesis}\label{hypo:optim_2}
    Per Hypothesis \ref{hypo:optim_1}, the points in the Cartesian plane where the independent agent switches the control input for the maxima of the optimization problem in (\ref{eq:optim_cost})-(\ref{eq:optim_bounds}) are given by $\underset{\mathbf{x}_s \in \mathcal{E}}{\max}~x_s$ and $\underset{\mathbf{x}_s \in \mathcal{E}}{\min}~x_s$. The points for the minima are given by $\underset{\mathbf{x}_s \in \mathcal{E}}{\max}~y_s$ and $\underset{\mathbf{x}_s \in \mathcal{E}}{\min}~y_s$.
\end{hypothesis}


The above hypotheses can be visualized using the simulation results shown in Figure \ref{fig:ellipse}, which depicts four distinct cases from the numerous simulations that were performed to analyze the extrema of the optimization problem in (\ref{eq:optim_cost})-(\ref{eq:optim_bounds}).

\begin{figure}[htb!]
\centering
\includegraphics[width = 0.65\textwidth]{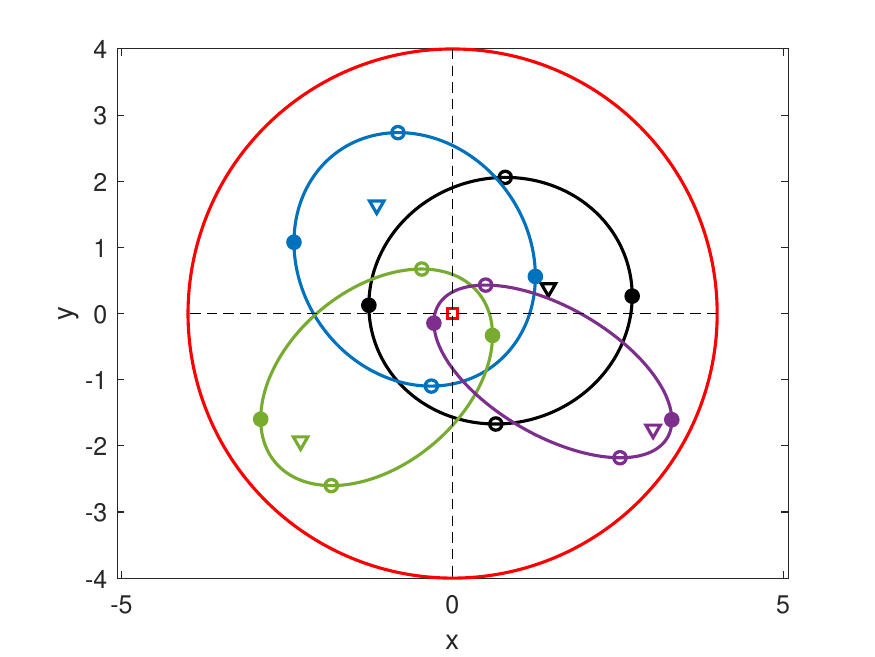}
\caption{Simulation results that visualize the extrema of the optimization problem in (\ref{eq:optim_cost})-(\ref{eq:optim_bounds}), per Hypotheses \ref{hypo:optim_1} and \ref{hypo:optim_2}.}
\label{fig:ellipse}
\end{figure}

The simulations corresponding to Figure \ref{fig:ellipse} were conducted with $\mathbf{x}_I(0) = [0,0]^\top$ (red square), $v_D = 1$, $v_I = 0.5$, $t = 8$.
The four points (blue, black, purple, and green inverted triangles) chosen within $\mathcal{R}_I(t)$ (interior of the red circle) represent a subset of the possible locations of independent agent at time $t$. 
In each case, the corresponding colored ellipse represents the locus of switch points $\mathcal{E}$, as given in (\ref{eq:switch_locus}), that the independent agent could execute with a single-switching control function to reach the selected point from $x_I(0)$.
For each ellipse, the integral in (\ref{eq:optim_cost}) is evaluated for 360 equally spaced points (in terms of angular displacement) on the ellipse to identify the maximum and minimum points. 
On each ellipse in Figure \ref{fig:ellipse}, the maxima points are indicated using filled circles, and the minima points are represented using empty circles.
Hence, it can be observed that the maxima points always correspond to the points with the maximum and the minimum horizontal coordinates on the ellipse.
Similarly, the minima points coincide with the points having the maximum and minimum vertical coordinates on the ellipse. As $\mathbf{x}_I(t)$ varies between the origin and the boundary of $\mathcal{R}_I(t)$, the locus of the switch points changes from circular to elliptical with increasing eccentricity.

\begin{figure}[htb!]
\centering
\includegraphics[width = 0.65\textwidth]{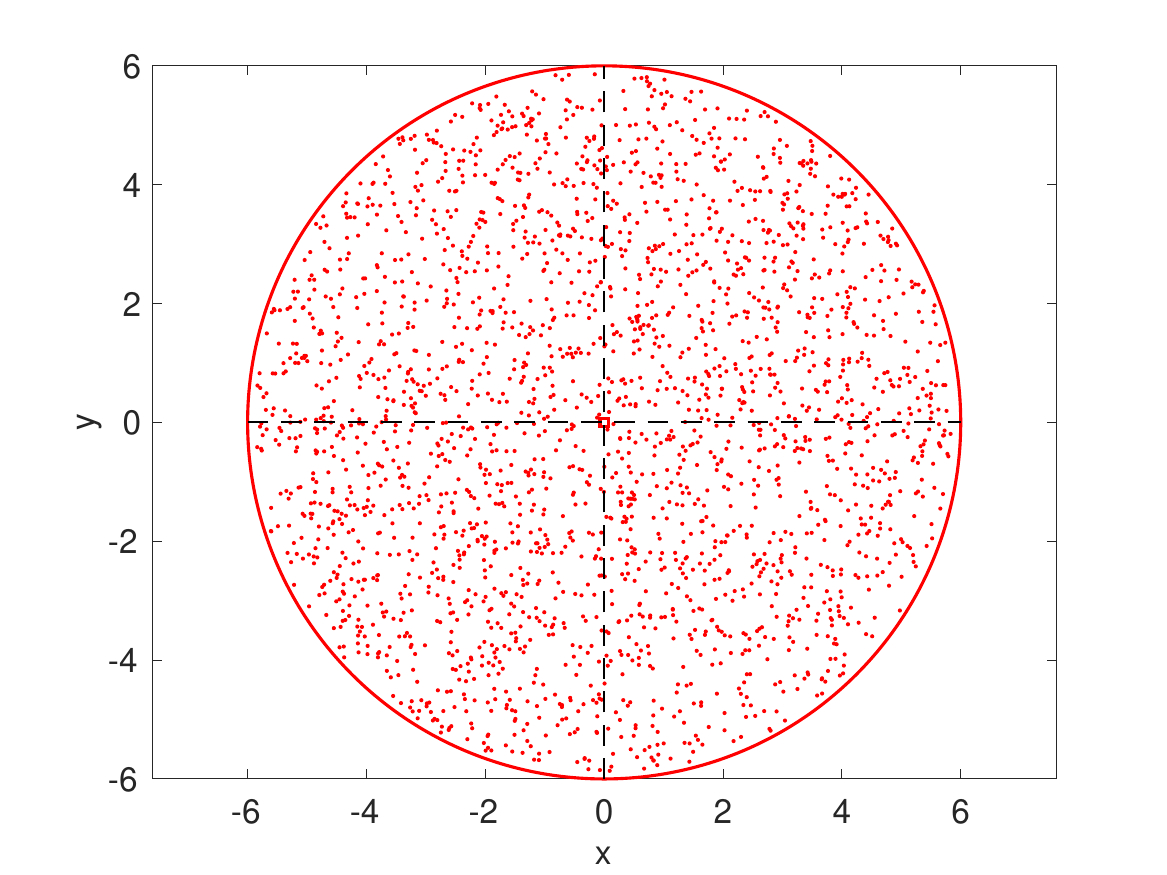}
\caption{Scatter plot representing the points from $\mathcal{R}_I(t)$ that are considered for Monte Carlo simulations.}
\label{fig:ellipse_scatter}
\end{figure}

To further validate Hypothesis \ref{hypo:optim_2}, Monte Carlo simulations were conducted. 
To this end, the simulation parameters for the independent agent are fixed as $\mathbf{x}_I(0) = [0,0]^\top$, $v_I = 1.2$, $t = 5$.
A total of 2,500 points are randomly chosen from the reachable set of the independent agent $\mathcal{R}_I(t)$, which is a circle with radius $6$, using the standard uniform distribution.
A scatter plot of the random points along with $\partial \mathcal{R}_I(t)$ is shown in Figure \ref{fig:ellipse_scatter}.
Four different speed ratios $v_I/v_D = 0.8, 0.6, 0.4, 0.2$, which correspond to $v_D = 1.5, 2, 3, 6$, respectively, are considered for the Monte Carlo simulations.
Consequently, a total of 10,000 distinct ellipses are analyzed to identify the extrema of the integral in (\ref{eq:optim_cost}).
For each ellipse, the extrema analytically computed using Hypothesis \ref{hypo:optim_2} is compared against the extrema that are obtained numerically.
The maximum root mean square error for the four extremum points across the 10,000 ellipses is found to be of the order of $10^{-15}$, thus validating the hypothesis.


\section{Conclusion} 
\label{sec:conclude}
This paper analyzes the ``reachable set" of an agent (termed the dependent agent) in instances where it follows another agent (termed the independent agent) using the constant bearing pursuit strategy. 
It is assumed that the speed of the dependent agent is greater than the speed of the independent agent.
Theoretical results are presented that analytically characterize the dependent reachable set for the instances where the dependent agent's reachable set does not fully engulf the diameter of the independent agent's reachable set.
It is shown that the theoretical results can be extended to the limiting case where the speeds of both agents are equal.
In instances where the dependent agent's reachable set entirely engulfs the diameter of the independent agent's reachable set, a tight bound for the dependent reachable set is provided along with a hypothesis supported by simulation-based empirical evidence.
The study of the dependent reachable set for the constant bearing pursuit strategy led us to examine a novel optimization problem, which resulted in more empirical evidence-based hypotheses related to the geometry of the constant bearing pursuit strategy and the ellipse.

\bibliographystyle{plain}
\bibliography{references}

\end{document}